\documentclass[11pt]{article}

\usepackage{amsmath,amsthm,nicefrac}
\usepackage{amsfonts, amstext}
\usepackage[upright=true]{derivative}
\usepackage{booktabs}

\usepackage{comment}

\usepackage{typearea}
\typearea{14}
\usepackage[font={small,it}]{caption}

\usepackage{setspace}
\usepackage[compact]{titlesec}
\usepackage{needspace}
\usepackage{enumitem}
\usepackage[breaklinks]{hyperref}
\hypersetup{colorlinks=true,%
    citebordercolor={.6 .6 .6},linkbordercolor={.6 .6 .6},%
    citecolor=blue,urlcolor=black,linkcolor=blue}

\usepackage[nameinlink,capitalize]{cleveref}
\Crefname{algocf}{Algorithm}{Algorithms}
\crefname{algocfline}{line}{lines}
\Crefname{ALC@unique}{Line}{Lines}
\Crefname{invariant}{Invariant}{Invariants}
\Crefname{claim}{Claim}{Claims}
\Crefname{subclaim}{Subclaim}{Subclaims}

\usepackage{epsfig}
\usepackage{amssymb}

\usepackage{mathrsfs}
\usepackage{xspace}
\usepackage{soul}
\usepackage{latexsym}
\usepackage{bm}

\usepackage{framed}

\usepackage[dvipsnames]{xcolor}
\definecolor{DarkGray}{rgb}{0.66, 0.66, 0.66}
\definecolor{DarkPowderBlue}{rgb}{0.0, 0.2, 0.6}
\definecolor{fluorescentyellow}{rgb}{0.8, 1.0, 0.0}

\usepackage[ruled,vlined,linesnumbered,algonl]{algorithm2e}
\SetEndCharOfAlgoLine{}
\SetKwComment{Comment}{\footnotesize$\triangleright$\ }{}

\SetCommentSty{mycommfont}

\usepackage{algorithmic}
\usepackage{mathtools}

\allowdisplaybreaks

\newcounter{note}[section]

\sethlcolor{fluorescentyellow}

\newcommand{\initOneLiners}{%
    \setlength{\itemsep}{0pt}
    \setlength{\parsep }{0pt}
    \setlength{\topsep }{0pt}
}

\newtheorem{theorem}{Theorem}
\newtheorem{lemma}[theorem]{Lemma}

\theoremstyle{definition}

\theoremstyle{remark}

\makeatletter
\renewcommand{\theinvariant}{(I\@arabic\c@invariant)}
\makeatother

\newcommand{\eps}{\varepsilon}

\providecommand{\half}{\frac{1}{2}}

\newcommand{\EE}{\mathbb E}

\newcommand{\SA}{\mathcal A}

\newcommand{\SD}{\mathcal D}

\newcommand{\eat}[1]{}

\usepackage{tikz}
\usetikzlibrary{arrows.meta,positioning,calc}

\DeclareMathOperator{\ALG}{ALG}
\DeclareMathOperator{\OPT}{OPT}

\title{The Competitive Ratio of Online Budgeted Allocation is Strictly Smaller than $1-1/e$}
\author
{Debmalya Panigrahi\thanks{Department of Computer Science, Duke University.
Email: {\tt debmalya@cs.duke.edu}. Supported in part by NSF grant CCF-2329230.}
\and Ian Zhang\thanks{Duke University. Email: {\tt ian.zhang@duke.edu}.}}
\date{October 2026}

\begin{document}
    \maketitle

    \begin{abstract}
        A well-known conjecture in online algorithms is that the (randomized) competitive ratio of online budgeted
        allocation is $1-1/e$. We refute this conjecture
        and show that every randomized online algorithm has competitive ratio at most
        $1 - 1/e - \eps$ for a (small) absolute constant $\eps > 0$.
    \end{abstract}

    \section{Introduction}\label{sec:intro}
    We study the online budgeted allocation problem. In this problem, a set of $M$ items arrive online and have to be assigned to a set of $N$ advertisers, where each advertiser $j$ has a non-negative budget $B_j \ge 0$ given offline. When item $i$ arrives online, each advertiser $j$ places a non-negative bid $b_{ij} \ge 0$ on the item. The online algorithm can assign item $i$ to any advertiser $j$ and earn a payment of $\min(b_{ij}, R_j)$, where $R_j$ is the remaining budget of advertiser $j$. Correspondingly, the remaining budget of advertiser $j$ is reduced by this payment. The goal of the online algorithm is to maximize the total revenue, which is the sum of payments for all items.

When all budgets are $1$ and all bids are either $0$ or $1$, online budgeted allocation is the same as online bipartite matching. For the latter problem, the tight (randomized) competitive ratio is $1-1/e$, achieved by the RANKING algorithm of Karp, Vazirani, and Vazirani \cite{KVV90} (see also \cite{GM08,BM08,DJK13}). The maximum bid-to-budget ratio (call it $\rho$) is $1$ for online bipartite matching. Motivated by applications in online advertising, Mehta, Saberi, Vazirani, and Vazirani \cite{MSVV07} considered the contrasting situation where $\rho \rightarrow 0$, and called this the AdWords problem. Generalizing a previous algorithm of Kalyanasundaram and Pruhs \cite{KP00} for $b$-matching, they obtained a deterministic algorithm for AdWords whose competitive ratio tends to $1-1/e$ as $\rho \rightarrow 0$ \cite{MSVV07} (see also \cite{BJN07}). This latter result also shows that for online budgeted allocation (with arbitrary $\rho$), there is a {\em fractional} algorithm with a competitive ratio of $1-1/e$. 

But, what about general bids and budgets? Motivated by the above results, 
a well-known conjecture in online algorithms says that the randomized competitive ratio of online budgeted allocation is $1-1/e$ (for any bids and budgets). However, progress toward this conjecture has been slow. The greedy algorithm that assigns every item to the advertiser that would maximize its payment for the item has a competitive ratio of $1/2$. Obtaining a better competitive ratio (using a randomized algorithm, since there is a simple upper bound of $1/2$ for deterministic algorithms) for the online budgeted allocation problem was a major open problem for more than a decade. In a recent breakthrough result, Huang, Zhang, and Zhang \cite{HZZ24} obtained a competitive ratio of $0.5016$, thereby slightly improving on the bound given by the greedy algorithm, but providing further evidence that the eventual competitive ratio might turn out to be $1-1/e$. Thereafter, this bound was improved by Chen, Huang, and Sun to $0.504$~\cite{CHS24}, and this is the best bound currently known for the problem. 

In this paper, we show that the competitive ratio of any (randomized) algorithm for online budgeted allocation is at
most $1-1/e-\eps$, for some absolute constant $\eps > 0$ (of the order of $10^{-4}$). This refutes the conjectured randomized competitive ratio of $1-1/e$. Furthermore, in conjunction with the algorithmic results of \cite{HZZ24,CHS24}, it shows that the competitive ratio lies strictly in the interval $(1/2, 1-1/e)$.

Our construction is quite simple: it uses a budget of $1$ for every advertiser, and only four different bid values: $0$, $1/2$, $3/4$, and $1$. We hope this simplicity will help in future research that tries to narrow the gap between the upper and lower bounds for this problem.

We state our main result below.

\begin{theorem}\label{thm:main}
    There are constants $\eps > 0$ and $N_0$ such that, for every integer $N \ge N_0$
    and every randomized online algorithm $\SA$, some fixed input $\sigma$
    on $N$ unit-budget advertisers with bids in $\{0, 1/2, 3/4, 1\}$ satisfies
    \[
        \EE[\ALG_{\SA}(\sigma)]
        \le \left(1 - \frac1e - \eps\right)\cdot \OPT(\sigma).
    \]
\end{theorem}

    \section{The Input Distribution}\label{sec:construction}
    We construct an input distribution for any sufficiently large $N$
that establishes the upper bound in \Cref{thm:main} on the expected competitive ratio of any deterministic algorithm.
By Yao's minimax principle~\cite{Yao77}, this implies the same upper bound
for randomized algorithms on a worst-case input.

Let us start by recalling the standard upper bound of $1-1/e$.
All advertisers have a budget of $1$.
The construction has $N$ {\em rounds}, where in each round,
an item is presented that has a bid value of $1$ from each remaining advertiser.
Once the algorithm assigns this item, one advertiser chosen uniformly at random is removed from the remaining set,
before continuing on to the next round.
When an advertiser is removed, all its bids on subsequent items are of value $0$.

Our new construction is similar in that there are multiple rounds, where advertisers are removed after rounds.
These rounds are organized into three {\em phases}: we call these the {\em first}, {\em middle}, and {\em final} phases.
The first and final phase are very similar to the standard construction, but the middle phase is different.

\paragraph{First phase.}
The first phase comprises $N/3$ rounds. (Assume w.l.o.g. that $N$ is a multiple of $3$.)
Each round has a single item with a bid of $3/4$ from each remaining advertiser.
After the algorithm assigns the item, an advertiser chosen uniformly at random is removed from the
remaining set, before continuing on to the next round. In essence, this is identical to the
classical construction, except that the bid values on the items are $3/4$ instead of $1$.
Note that the number of remaining advertisers at the end of this phase is
$M := 2N/3$.

\paragraph{Middle phase.}
Fix an even number $d \ge 8$, to be chosen at the end of the proof.
Arbitrarily partition the remaining $M$ advertisers into $M/d$ groups of size $d$ each,
and further split each group into two halves of size $d/2$ each.
Order the groups arbitrarily and denote the subgroups $A_i, B_i$ for the $i$th group.
(Assume w.l.o.g. that $M$ is a multiple of $d$. Note that $d$ is chosen to be
even; hence, $d/2$ is an integer.)

There are $M/d$ rounds in the middle phase. The $i$th round comprises
three consecutive items with the following bids (all other advertisers bid $0$):
\begin{itemize}
    \item[-] Item $1$: All $d$ advertisers in the $i$th group bid $1$ for the item.
    \item[-] Item $2$: All $d$ advertisers in the $i$th group bid $1/2$ for the item.
    \item[-] Item $3$: Choose one of the two subgroups $A_i, B_i$, with probability $1/2$ each.
    All $d/2$ advertisers in the chosen subgroup bid $1/2$ for the item.
\end{itemize}

After all items have been assigned, choose an advertiser uniformly at random
from each subgroup $A_i, B_i$ and remove them from the remaining set.
(For notational convenience, we perform this removal of advertisers only
after all the $M/d$ rounds have been completed, i.e., all $3M/d$ items
have been assigned.)

Note that this removes a total of $2M/d$ advertisers in the middle phase.
We denote the number of remaining advertisers
$R \coloneqq M(1 - 2/d)$.

\paragraph{Final phase.}
The final phase comprises $R$ rounds. This phase is identical to the classical construction.
Each round has a single item with a bid of $1$ from each remaining advertiser.
After the algorithm assigns the item, an advertiser chosen uniformly at random is removed from the
remaining set, before continuing on to the next round.

We denote the overall input distribution constructed across the three phases $\SD_{N,d}$.

    \section[Proof of Theorem 1]{Proof of \Cref{thm:main}}\label{sec:analysis}

    \subsection{Optimal Revenue}\label{sec:opt}
    We start by establishing the revenue $\OPT$ of an optimal solution.

\begin{lemma}\label{lem:opt}
    Every input in the support of $\SD_{N,d}$ has $\OPT = 11N/12$.
\end{lemma}
\begin{proof}
    In the first phase, assign each item to the advertiser removed immediately afterward.
    This collects $3/4$ per item and leaves every surviving advertiser's budget full
    at the beginning of the middle phase.
    During the middle phase, for a group where advertisers in $A_i$ (resp., $B_i$) bid for item $3$,
    give item $1$ to the advertiser to be removed from $B_i$ (resp., $A_i$),
    and both items $2$ and $3$ to the advertiser to be removed from $A_i$ (resp., $B_i$).
    This collects $2$ and leaves all advertisers surviving into the final phase with full budgets.
    In the final phase, again assign each item to the advertiser removed immediately afterward, collecting $1$.
    The total is $(3/4) \cdot N/3 + 2M/d + R = 11N/12$.
    Note that this allocation collects the maximum bid on every item, which is also an upper bound on any allocation.
\end{proof}

    \subsection{A Potential Function to Bound the Algorithm's Revenue}\label{sec:potential}
    We are now left to upper bound the expected revenue of any deterministic online algorithm
on an input drawn from $\SD_{N, d}$.

Fix an arbitrary deterministic online algorithm and let $\ALG$ denote the algorithm's revenue.
We will analyze the algorithm using a potential function $\varphi$ that changes through the course of the algorithm.
Let the initial and final values of $\varphi$ be $\varphi_i$ and $\varphi_f$ respectively.
During the course of the algorithm, the value of $\varphi$ changes due to two types of 
events: (a) when the algorithm assigns an item to an advertiser, and (b) when one or more advertisers
is removed from the set of remaining advertisers. It will turn out that the expected
value of the potential function $\varphi$ does not change for events of type (b). 
We denote the change in $\varphi$ due to the algorithm's assignment of item $t$
to be $\Delta_t$.

Our potential will satisfy $\varphi_f = \ALG$.
Since $\EE[\varphi_f] = \varphi_i + \sum_{t}\EE[\Delta_t]$ (recall that $\varphi$
does not change in expectation due to removal of an advertiser), 
this allows us to bound $\EE[\ALG]$ by the sum of $\varphi_i$ and the
expected values of the local changes $\EE[\Delta_t]$.
Throughout the proof, unless conditioned otherwise, expectations are over the input distribution $\SD_{N,d}$.

Now, we are ready to define the potential function $\varphi$.
First, define the {\em capped} remaining budget $v(r)$ of an advertiser whose remaining 
budget is $r$ to be
\[
    v(r) = \max\left(r - \frac{1}{12}, 0\right).
\]
Let $\ALG_t$ be the total revenue collected by the algorithm after assigning items $1, 2, \ldots, t$,
and let $\bar{v}_t$ be the average value of $v(r)$ over the remaining advertisers at this stage.
(By convention, when there is no remaining advertiser at the end, then $\bar{v}_t = 0$.)
Set $H \coloneqq N/e$ to be a fixed scaling factor.
Define the potential
\begin{equation}\label{eq:pot}
    \varphi(t) = \ALG_t + H\cdot \bar{v}_t.
\end{equation}

(As an aside, we show in \Cref{app:upper-triangular} that the same potential with $v(r)$ replaced by $r$
recovers the $1 - 1/e$ upper bound for the standard upper-triangular instance.)

We now bound the expected revenue of the algorithm using the changes in potential $\Delta_t$.

\begin{lemma}
    \label{lem:master}
    The expected revenue of the algorithm satisfies
    \begin{equation}\label{eq:master}
        \EE[\ALG] = \frac{11H}{12} + \sum_t\EE[\Delta_t].
    \end{equation}
\end{lemma}
\begin{proof}
    Let us first note the initial and final values of the potential function $\varphi$.
    Initially, $v(r) = 11/12$ for each advertiser. Therefore, the initial value is $\varphi_i = 11H/12$.
    At the end of the algorithm, the set of remaining advertisers is empty and we define $\bar{v}$ to be $0$.
    Therefore, the final value is $\varphi_f = \ALG$.

    In order to claim that $\EE[\varphi_f] = \varphi_i + \sum_{t}\EE[\Delta_t]$,
    we now show that each removal of advertiser(s) does not change $\varphi$ in expectation.
    The removal of an advertiser can only affect $\bar{v}_t$ in \eqref{eq:pot}, 
    so we show that $\bar{v}_t$ is preserved.
    In the first and final phases, the next advertiser removed is uniform among the advertisers still present,
    and removing a uniformly chosen advertiser preserves the average capped budget $\bar{v}_t$ in expectation.
    For removals in the middle phase, each removed advertiser is uniformly chosen
    from their respective halves $A_i$ and $B_i$.
    Therefore, all advertisers have the same probability $2/d$ of being removed,
    and hence, the expectation of $\bar{v}_t$ is again unchanged.
\end{proof}

We now write the changes $\Delta_t$ in terms of the change to the revenue collected by $\ALG$
for item $t$, and the change in $\bar{v}_t$.
For the first part, let $p_t$ be the revenue collected by $\ALG$ for item $t$.
To define the change in $\bar{v}_t$ precisely, let $s_t$ denote
the number of remaining advertisers just before assigning item $t$.
Suppose $r$ is the remaining budget of the advertiser that item $t$ is assigned to.
Define $w_t \coloneqq v(r) - v(r - p_t)$ to be the decrease in the advertiser's capped budget caused by this assignment.
We have
\begin{equation}
    \Delta_t = p_t - \frac{H}{s_t}\cdot w_t.
\end{equation}

    \subsection{Change in Potential for the Three Phases}\label{sec:local}
    It now suffices to bound the local changes to the potential in each of the three phases.
We begin by bounding $\Delta_t$ in the first phase.

\begin{lemma}\label{lem:first}
    For every item $t$ arriving in the first phase, we have
    \[
        \Delta_t \le \frac34\left(1 - \frac{H}{s_t}\right).
    \]
\end{lemma}
\begin{proof}
    Since the only positive bids in the first phase have value $3/4$,
    the remaining budgets of the advertisers must belong to $\{0, 1/4, 1\}$.
    If the algorithm assigns item $t$ to an advertiser with budget $1$,
    we have $p_t = w_t = 3/4$ which gives the desired bound on $\Delta_t$.

    Now recall that $M = 2N/3$ is the number of advertisers remaining after the first phase.
    Assigning $t$ to an advertiser with budget $1/4$ gives $p_t = 1/4$ and $w_t = 1/6$.
    Its change is smaller because
    \[
        \frac34\left(1 - \frac{H}{s_t}\right)
        - \left(\frac14 - \frac{H}{6s_t}\right)
        = \half - \frac{7H}{12s_t} > 0,
    \]
    since $s_t > M$ and $H/M = 3/(2e) < 3/4$.
    Finally, assigning $t$ to an advertiser with budget $0$ or to a removed advertiser gives a change of $0$,
    whereas the stated bound on $\Delta_t$ is positive because $s_t > M > H$.
\end{proof}

We now consider the middle phase.

\begin{lemma}\label{lem:middle}
    Let items $t, t+1, t+2$ be the items corresponding to a group in the middle phase.
    Then,
    \begin{equation}\label{eq:group}
        \EE[\Delta_t + \Delta_{t + 1} + \Delta_{t + 2}] \le
        2\left(1 - \frac{11\alpha}{12}\right) - \frac{\alpha}{24}, \text{ where } \alpha = \frac{H}{M}.
    \end{equation}
    The expectation is over the random choice of which subgroup bids for item $3$.
\end{lemma}
\begin{proof}
    Throughout the assignment of these items, there are $M$ advertisers present,
    and hence, a payment of $p$ consuming capped budget $w$ changes the potential by $p - \alpha w$.

    Now, consider only the group corresponding to these items.
    All remaining advertiser budgets in this group are in $\{0, 1/4, 1\}$.
    The first item has a bid of $1$. We do a case analysis for assigning 
    this item to advertisers with the three possible remaining budgets.
    Assigning it to an advertiser with budget $1$ gives a change $1 - 11\alpha/12$.
    Assigning it to an advertiser with budget $1/4$ gives $1/4 - \alpha/6$, which 
    is smaller than $1 - 11\alpha/12$ by $(3/4)(1 - \alpha) > 0$.
    Finally, a $0$ payment changes the potential by $0$, which is also
    less than $1 - 11\alpha/12$.
    Thus, in all three cases, we have $\Delta_t \le 1 - 11\alpha/12$.

    Now, we consider the second and third items. First, note that after assigning
    the first item, the remaining advertiser budgets still belong to $\{0, 1/4, 1\}$.
    We again do a case analysis for the second item.
    Assigning the item to an advertiser with remaining budget $1$ changes 
    the potential by $(1 - \alpha)/2$. 
    Assigning the item to an advertiser with remaining budget $1/4$ changes 
    the potential by $1/4 - \alpha/6$, which is smaller than the 
    previous case since
    \[
        \frac{1 - \alpha}{2} - \left(\frac14 - \frac\alpha6\right)
        = \frac14 - \frac\alpha3 > 0.
    \]
    Finally, a $0$ payment is also smaller than $(1 - \alpha)/2$.
    Thus $\Delta_{t + 1} \le (1 - \alpha)/2$.

    The second item creates at most one advertiser with remaining budget $1/2$.
    All other advertisers still have remaining budgets in $\{0, 1/4, 1\}$.
    Recall that the third item is bid for by either $A_i$ or $B_i$,
    chosen uniformly.
    Thus, if there exists an advertiser with remaining budget $1/2$,
    it lies on the side that bids for the third item with probability $1/2$.
    By the calculation for $\Delta_{t+1}$,
    assigning the third item to an advertiser with budget $1$ or $1/4$ still
    produces a change of at most $(1 - \alpha)/2$.
    Assigning it to an advertiser with budget $1/2$ instead produces a change of
    \[
        \half - \alpha\cdot  v(1/2) = \frac{1 - \alpha}{2} + \frac\alpha{12}.
    \]
    The extra $\alpha/12$ is available only if the unique advertiser with 
    remaining budget $1/2$
    lies on the side that bids for the third item, which occurs with probability at most $1/2$.
    Consequently, $\EE[\Delta_{t + 2}] \le (1 - \alpha)/2 + \alpha/24$.
    Adding the three bounds gives
    \[
        \EE[\Delta_t + \Delta_{t + 1} + \Delta_{t + 2}] \le
        \left(1 - \frac{11\alpha}{12}\right)
        + \frac{1 - \alpha}{2}
        + \left(\frac{1 - \alpha}{2} + \frac\alpha{24}\right)
        = 2\left(1 - \frac{11\alpha}{12}\right) - \frac\alpha{24}. \qedhere
    \]
\end{proof}

Finally, we consider the final phase.
\begin{lemma}\label{lem:final}
    For every item $t$ arriving in the final phase, we have
    \[
        \Delta_t \le \frac34\left(1 - \frac {H}{s_t}\right)_+  + \frac14\left(1 - \frac23\cdot\frac{H}{s_t}\right)_+,
        \text{ where } z_+ = \max\{z, 0\}.
    \]
    
\end{lemma}
\begin{proof}
    Let $t$ be any item arriving in the final phase.
    Since all bids are $1$ or $0$, a positive payment empties its recipient's budget $r$.
    Write $r = u + \ell$, where $u = (r - 1/4)_+$ and $\ell = \min\{r, 1/4\}$.
    Note that $0 \le u \le 3/4$ and $0 \le \ell \le 1/4$.
    Since all remaining advertiser budgets are in $\{0, 1/4, 1/2, 1\}$,
    we have $v(r) = u + (2/3)\ell$.
    This implies
    \[
        \Delta_t = \left(1 - \frac{H}{s_t}\right)u
        + \left(1 - \frac{2}{3}\cdot \frac{H}{s_t}\right)\ell
        \le \frac34\left(1 - \frac{H}{s_t}\right)_+
        + \frac14\left(1 - \frac{2}{3}\cdot\frac{H}{s_t}\right)_+.
    \]
    A payment of $0$ does not change the potential, and satisfies the bound
    above, which is non-negative.
\end{proof}

    \subsection{Adding the Changes}\label{sec:finish}
    It remains to bound the sum of changes $\Delta_t$ across the three phases.
Recall that in the first and final phases, successive items arrive with
$N$, $N - 1$, $\dots$, $M + 1$ advertisers and $R$, $R - 1$, $\dots$, $1$ advertisers remaining, respectively.
Therefore, using \Cref{lem:first,lem:middle,lem:final} in \eqref{eq:master}, we get
\begin{equation}\label{eq:sum}
    \EE[\ALG] \le \frac{11H}{12}
    + \sum_{s = M + 1}^{N}\frac34\left(1 - \frac Hs\right)
    + \Delta_{\mathrm{mid}} + \sum_{s = 1}^{R}F(s),
\end{equation}
where $\Delta_{\mathrm{mid}}$ denotes the total expected sum of $\Delta_t$ for the entire middle phase, and
\begin{equation}\label{eq:F}
    F(s) = \frac34\left(1 - \frac Hs\right)_+ + \frac14\left(1 - \frac23\cdot \frac{H}{s}\right)_+.
\end{equation}

Replacing the two sums in \eqref{eq:sum} with integrals, we have 
\begin{equation}\label{eq:integrals}
    \EE[\ALG] \le \frac{11H}{12}
    + \frac34\int_M^N\left(1 - \frac Hx\right) \odif{x}
    + \Delta_{\mathrm{mid}} + \int_0^R F(x) \odif{x} + 2,
\end{equation}
since $1-H/x \le 1$ and $F(x) \le 1$ and are non-decreasing in $x$.

For the remainder of the proof, we will bound $\EE[\ALG]$ solely using \eqref{eq:integrals}.
We first show that if the middle phase were to be omitted, i.e.,
if the final phase started with $M$ advertisers instead of $R$,
then we would get exactly the original bound of $(1-1/e)\OPT$.
We will then show in the subsequent lemma that adding the middle phase decreases our bound on $\EE[\ALG]$
by a constant fraction of $\OPT$.

\begin{lemma}
    \label{lem:baseline}
    We have
    \begin{equation}\label{eq:baseline}
        \frac{11H}{12} + \frac34\int_M^N\left(1 - \frac Hx\right) \odif{x}
        + \int_0^M F(x) \odif{x}
        = \left(1 - \frac1e\right)\OPT.
    \end{equation}
\end{lemma}
\begin{proof}
    We first evaluate $\int_0^M F(x) \odif{x}$
    using the definition of $F$ in~\eqref{eq:F}.
    The factor $(1 - H/x)_+$ is zero for $0 < x \le H$ and equals $1 - H/x$ for $x \ge H$.
    Similarly, $(1 - 2H/(3x))_+$ is zero for $0 < x \le 2H/3$
    and equals $1 - 2H/(3x)$ for $x \ge 2H/3$.
    Since $M > H$, integrating these two terms gives
    \[
        \int_0^M F(x) \odif{x}
        = \frac34\int_H^M\left(1 - \frac Hx\right) \odif{x}
        + \frac14\int_{2H/3}^M\left(1 - \frac{2H}{3x}\right) \odif{x}.
    \]
    The integral over $[H, M]$ on the right has the same integrand and coefficient
    as the integral over $[M, N]$ in~\eqref{eq:baseline}.
    Adding them gives a single integral over $[H, N]$ with coefficient $3/4$.
    We also write $11H/12 = (3/4)H + (1/4)(2H/3)$
    so that we can group the terms with coefficients $3/4$ and $1/4$.
    These rewritings give
    \begin{align*}
        &\frac{11H}{12} + \frac34\int_M^N\left(1 - \frac Hx\right) \odif{x}
        + \int_0^M F(x) \odif{x} \\
        &\qquad= \frac34\left[H + \int_H^N\left(1 - \frac Hx\right) \odif{x}\right]
        + \frac14\left[\frac{2H}{3}
                     + \int_{2H/3}^M\left(1 - \frac{2H}{3x}\right) \odif{x}\right] \\
        &\qquad= \left(1 - \frac1e\right)\left(\frac34 N + \frac14 M\right)
        = \left(1 - \frac1e\right)\OPT.
    \end{align*}
    For the second equality,
    the first bracket evaluates to $N - H\ln(N/H) = N - H$, since $N/H = e$.
    The second bracket similarly evaluates to $M - 2H/3$, since $M/(2H/3) = e$.
    Substituting $H = N/e$ and $2H/3 = M/e$ gives the common factor $1 - 1/e$.
    The last equality uses $M = 2N/3$ and $\OPT = 11N/12$.
\end{proof}

Applying~\Cref{lem:baseline} to~\eqref{eq:integrals}, we get
\begin{equation}\label{eq:comparison}
    \EE[\ALG] \le \left(1 - \frac1e\right)\OPT
    + \Delta_{\mathrm{mid}} - \int_{R}^M F(x) \odif{x} + 2.
\end{equation}
It remains to show that the middle contribution $\Delta_{\mathrm{mid}}$ is smaller than the integral
$\int_{R}^M F(x) \odif{x}$ it replaces.

First, we calculate $\Delta_{\mathrm{mid}}$.
The middle phase has $M/d$ groups, with $M$ advertisers present throughout the phase.
Summing the bound in~\Cref{lem:middle} over these groups and substituting $\alpha = H/M$
gives the following upper bound on the total expected change in the middle phase:
\begin{equation}\label{eq:mid}
    \Delta_{\mathrm{mid}} \le \frac Md\left[2\left(1 - \frac{11H}{12M}\right) - \frac{H}{24M}\right].
\end{equation}

\begin{lemma}\label{lem:middle-comparison}
    For every sufficiently large even constant $d$, there is a constant $\eps > 0$, independent of $N$, such that
    \[
        \Delta_{\mathrm{mid}} - \int_R^M F(x) \odif{x} < -\eps N.
    \]
\end{lemma}
\begin{proof}
    Since $M > H$, we have $F(M) = 1 - 11H/(12M)$.
    Using $M - R = 2M/d$, we can write~\Cref{eq:mid} as
    \[
        \Delta_{\mathrm{mid}} \le (M - R)\left(F(M) - \frac{H}{48M}\right).
    \]
    Since $F$ is non-decreasing, we have $F(x) \ge F(R)$ throughout $[R, M]$, so
    \[
        \Delta_{\mathrm{mid}} - \int_R^M F(x) \odif{x}
        \le (M - R)\left(F(M) - F(R) - \frac{H}{48M}\right).
    \]
    By~\eqref{eq:F}, the values $F(M)$ and $F(R)$ depend only on
    $H/M = 3/(2e)$ and $R/M = 1 - 2/d$, and not on $N$.
    As $d$ grows, $R/M$ tends to $1$, so continuity gives $F(M) - F(R) \to 0$.
    In contrast, $H/(48M)$ is a fixed positive constant.
    Fixing a sufficiently large even $d$ therefore makes the expression in parentheses
    a negative constant independent of $N$.
    Since $M - R = 4N/(3d)$, the claim follows.
    (We note that choosing $d=90$, one can show the lemma statement for $\eps = 1/(4320e)$).
\end{proof}

\Cref{thm:main} now follows from \Cref{lem:baseline,lem:middle-comparison}.

    \section*{AI Usage}
    We used the publicly available GPT-6 Astra model to obtain
    the first construction that gave an upper bound below $1-1/e$.
    However, this construction and its analysis were extremely complicated.
    Over a period of time, we progressively changed the construction to simplify it to its current form.

    \appendix

    \section{The Upper-Triangular Bound}\label{app:upper-triangular}
    In this section, we present a proof of the classical $1 - 1/e$ upper bound
using the upper-triangular construction of Karp, Vazirani, and Vazirani~\cite{KVV90}.
This proof uses a potential that has the same form as in~\Cref{sec:potential}.

We briefly restate the standard upper bound construction.
Start with $N$ advertisers, each with budget $1$.
The construction has $N$ {\em rounds}, where in each round,
an item is presented that has a bid value of $1$ from each remaining advertiser.
Once the algorithm assigns this item, one advertiser chosen uniformly at random is removed from the remaining set,
before continuing on to the next round.
When an advertiser is removed, all its bids on subsequent items are of value $0$.
An offline algorithm can assign each item to the advertiser removed immediately afterward
and collect a revenue of $1$ in every round, so every input in this distribution
(which we call $\SD_{\rm KVV}$) has $\OPT = N$.

\begin{lemma}
    \label{lem:upper-triangular}
    On the input distribution $\SD_{\rm KVV}$, every deterministic online algorithm satisfies
    \[
        \EE[\ALG] \le \left(1 - \frac1e\right)N + 1.
    \]
\end{lemma}
\begin{proof}
    Fix a deterministic online algorithm.
    Let $\ALG_t$ be the revenue collected from the first $t$ items.
    Let $H := N/e$. We use the potential function
    \[
        \varphi(t) = \ALG_t + H \cdot \bar{r}_t,
    \]
    where $\bar{r}_t$ denotes the average value of the remaining budget
    over the remaining advertisers, immediately after assigning item $t$.
    Initially, $\ALG_0 = 0$ and $\bar{r}_0 = 1$, so $\varphi_i = H$.
    At the end, $\ALG_N = \ALG$ and $\bar{r}_N = 0$, so $\varphi_f = \ALG$.
    We therefore bound $\ALG$ by the expected sum of changes to the potential.

    For item $t$, let $s_t$ be the number of advertisers present before its assignment
    and let $p_t$ be the payment collected.
    Every remaining budget is either $0$ or $1$, so $p_t \in \{0, 1\}$.
    The assignment increases revenue by $p_t$ and decreases
    the average remaining budget by $p_t/s_t$.
    The change in $\varphi$ at round $t$, before the removal of an advertiser, is therefore
    \begin{equation}
    \label{eq:triangular}
        \Delta_t = p_t\left(1 - \frac{H}{s_t}\right)
        \le \left(1 - \frac{H}{s_t}\right)_+,
    \end{equation}
    where $z_+ = \max\{z, 0\}$.
    Since each advertiser removed is uniform among the remaining advertisers,
    this does not change $\bar{r}$, and therefore $\EE[\varphi_f] = \varphi_i + \sum_{t}\EE[\Delta_t]$.
    Summing~\eqref{eq:triangular}, once for each $s_t = N, N - 1, \dots, 1$, gives
    \begin{align*}
        \EE[\ALG]
        &\le H + \sum_{s = 1}^{N}\left(1 - \frac Hs\right)_+
        \le H + \int_H^N\left(1 - \frac Hx\right)\odif{x} + 1 \\
        &= N - H\ln(N/H) + 1
        = \left(1 - \frac1e\right)N + 1.\qedhere
    \end{align*}
\end{proof}

        {\small
    \bibliographystyle{alphaurl}
    \bibliography{ian}
    }
\end{document}